\documentclass{amsart}
     
\newtheorem{theorem}{Theorem}[section]

\theoremstyle{definition}
\newtheorem{definition}[theorem]{Definition}

\theoremstyle{remark}
\numberwithin{equation}{section}

\def\a{{\alpha}}

\def\g{{\gamma}}

\def\z{{\zeta}}

\def\l{{\lambda}}

\def\na{{\nabla}}

\def\l{{\lambda}}

\def\ttt{{\theta}}

\def\G{{\Gamma}}
\def\O{{\Omega}}
\def\D{{\Delta}}

\def\bb{{{\mathcal B}}}
\def\cc{{{\mathcal C}}}

\def\uu{{{\mathcal U}}}

\def\uu{{{\mathcal U}}}

\def\uuu{{{\bf u}}}

\def\vvv{{{\bf v}}}
\def\xxx{{{\bf x}}}

\def\qqq{{{\bf q}}}

\def\R{{{\bf R}^1}}
\def\RR{{{\bf R}^2}}
\def\RRR{{{\bf R}^3}}
\def\RN{{{\bf R}^N}}

\def\UUU{{{\bf U}}}

\def\9{{\ \hbox{in}\ \O}}
\def\1{{\ \hbox{on}\ \G_1}}
\def\2{{\ \hbox{on}\ \G_2}}
\def\3{{\ \hbox{on}\ \G_3}}
\def\0{{\ \hbox{on}\ \G}}

\def\pa{{\partial}}
\def\pp{{\parallel}}

\begin{document}

\title[Functional solutions in heat and mass transfer]{Functional solutions for problems of heat and mass transfer}
\author{Giovanni Cimatti}
\address{Department of Mathematics, Largo Bruno
  Pontecorvo 5, 56127 Pisa Italy}
%\curraddr{Department of Mathematics and Statistics,
%Case Western Reserve University, Cleveland, Ohio 43403}
\email{cimatti@dm.unipi.it}
%\thanks{nume file: Airy.tex, dersec.eps, et.eps}

%\author{Author Two}
%\address{Mathematical Research Section, School of Mathematical Sciences,
%Australian National University, Canberra ACT 2601, Australia}

%    General info

%\date{January 1, 2001 and, in revised form, June 22, 2001.}



\begin{abstract}
We prove the existence and, in certain cases, the uniqueness of functional solutions for two boundary value problems of systems of P.D.E. in divergence form motivated by problems of heat and mass transfer. If $\cc_F$ and $\cc$ denote respectively the set of functional and classical solutions of these problems we settle, in simple cases, the question if $\cc_F=\cc$.
\end{abstract}

\maketitle

\section{Introduction}
In this paper we study two problems of heat and mass transfer modeled by systems of non-linear partial differential equations in divergence form with constant boundary conditions on different part of the boundary of an open and bounded subset of $\RRR$. In Sections 2 and 3 we treat the case of a purely molecular flow. In Section 4 we consider the macroscopic flow in a porous medium according to the Darcy's law with the conductivities depending on both temperature and concentration. We adopt the point of view of functional solutions \cite{GC}, \cite{GC1}. This permits to reformulate the boundary value problem with a non-standard one-dimensional two point problem for a system of O.D.E. coupled with a mixed problem for the laplacian. We present various results of existence and uniqueness of solutions.

\section{Purely molecular flow}
If in a fluid macroscopic motion is absent, the density of heat flow $\qqq_h$ and the density of mass flow $\qqq_m$ are related with the gradient of the temperature and the gradient of the concentration by the equations

\begin{equation}
\label{1_2}
\qqq_h=a_{11}\na u_1+a_{12}\na u_2,\quad \qqq_m=a_{21}\na u_1+a_{22}\na u_2,
\end{equation}
where $u_1$ denotes the temperature and $u_2$ the concentration. The kinetic coefficients $a_{ij}$ are assumed to depend on $u_1$ and $u_2$. No sources of mass or heat are present. Hence

\begin{equation}
\label{1_3}
\na\qqq_h=0,\quad \na\qqq_m=0.
\end{equation}
The flow takes place in a region represented by a subset $\O$ of $\RRR$. The boundary of $\O$ consists of three parts $\G_1$, $\G_2$ and $\G_3$. On $\G_1$ and $\G_3$ the temperature and the concentration are prescribed with different constant values. $\G_3$ is thermally insulated and no transport of mass occurs there. In view of (\ref{1_3}) and (\ref{1_2}) we obtain for the determination of $u_1(\xxx)$, $u_2(\xxx)$, $\xxx=(x_1,x_2,x_3)\in\O$ the non-linear boundary value problem \footnote{The constant temperature on $\G_1$ is taken as the zero value of an empirical scale of the temperature. The same is true for the concentration.}

\begin{equation}
\label{1_4}
\na\cdot\big[a_{11}(u_1,u_2)\na u_1+a_{12}(u_1,u_2)\na u_2\big]=0\quad\9
\end{equation}

\begin{equation}
\label{2_4}
u_1=0\quad\1,\quad\frac{\pa u_1}{\pa n}=0\quad\2,\quad u_1=u_1^*\quad\3
\end{equation}

\begin{equation}
\label{3_4}
\na\cdot\big[a_{21}(u_1,u_2)\na u_1+a_{22}(u_1,u_2)\na u_2\big]=0\quad\9
\end{equation}

\begin{equation}
\label{4_4}
u_2=0\quad\1,\quad\frac{\pa u_2}{\pa n}=0\quad\2,\quad u_2=u_2^*\quad\3.
\end{equation}
We assume the boundary of $\O$ so regular that the mixed problem for the laplacian

\begin{equation}
\label{2_3}
\D z=0\ \9,\quad z=0\ \1,\quad \frac{\pa z}{\pa n}=0\quad \2,\quad z=1\quad \3
\end{equation}
has one and only one classical solution. Let $\cc$ be the set of classical solutions of (\ref{1_4})-(\ref{4_4}). In this paper we are interested in a subset $\cc_F$ of $\cc$. The solutions in $\cc_F$ are defined as follows

\begin{definition}
Let $z(\xxx)$ be the solution of problem (\ref{2_3}). A solution $(u_1(\xxx), u_2(\xxx))\in\cc$ belongs to $\cc_F$ if two functions $U_1(z)$ and $U_2(z)$, both belonging to $C^1([0,1])$, exist such that

\begin{equation}
\label{1_5}
u_1(\xxx)=U_1(z(\xxx)),\quad u_2(\xxx)=U_2(z(\xxx)).
\end{equation}
\end{definition}
In the next Theorem we prove that all the functional solutions of the problem (\ref{1_4})-(\ref{4_4}) are known if the solution of problem (\ref{2_3}) is at our disposal and we can solve the non-standard two point problem

\begin{equation}
\label{1_6}
a_{11}(U_1(z),U_2(z))\frac{dU_1}{dz}+a_{12}(U_1(z),U_2(z))\frac{dU_2}{dz}=\g_1
\end{equation}

\begin{equation}
\label{2_6}
a_{21}(U_1(z),U_2(z))\frac{dU_1}{dz}+a_{22}(U_1(z),U_2(z))\frac{dU_2}{dz}=\g_2
\end{equation} 

\begin{equation}
\label{3_6}
U_1(0)=0,\quad U_1(1)=u_1^*,\quad U_2(0)=0,\quad U_2(1)=u_2^*,
\end{equation}
where $u_1^*$, $u_2^*$ are the constant boundary data of the problem (\ref{1_4})-(\ref{4_4}). In a sense the equations (\ref{1_6})-(\ref{3_6}) contain the non-linear part of problem (\ref{1_4})-(\ref{4_4}) and the linear problem (\ref{2_3}) reflects the geometric part.

\begin{theorem}
Let $(U_1(z),U_2(z),\g_1,\g_2)$ be a solution of (\ref{1_6})-(\ref{3_6}) and $z(\xxx)$ the solution of (\ref{2_3}). Define

\begin{equation}
\label{1_7}
u_1(\xxx)=U_1(z(\xxx)),\quad u_2(\xxx)=U_2(z(\xxx)).
\end{equation}
Then $(u_1(\xxx),u_2(\xxx))$ is a solution of the problem (\ref{1_4})-(\ref{4_4}). Vice-versa, let $(u_1(\xxx),u_2(\xxx))$ be a functional solution of (\ref{1_4})-(\ref{4_4}) and $U_1(z)$, $U_2(z)$ the two functions entering in the Definition 2.1. Then two constants $\g_1$ and $\g_2$ exist such that (\ref{1_6})-(\ref{3_6}) holds.
\end{theorem}

\begin{proof}
 We have for the functions defined in (\ref{1_7}), $u_1=u_2=0\1$, $u_1=u_1^*\3$ and $u_2=u_2^*\3$. Moreover, $\frac{\pa u_1}{\pa n}=U'\frac{\pa z}{\pa n}=0\ \2$ and, similarly,  $\frac{\pa u_2}{\pa n}=0\ \2$. On the other hand, $\na u_1=U_1'(z)\na z$, $\na u_2=U_2'(z)\na z$. Substituting in the left hand side of (\ref{1_4}) we have, in view of (\ref{2_3}) and (\ref{1_7}),

\begin{equation*}
%\label{1_10}
\na\cdot\big[a_{11}(u_1(\xxx),u_2(\xxx))\na u_1+a_{12}(u_1(\xxx),u_2(\xxx))\na u_2\big]=
\end{equation*}
\begin{equation*}
%\label{1_10}
\na\big[a_{11}(U_1(z(\xxx)),U_2(z(\xxx))U_1'(z(\xxx))+a_{12}(U_1(z(\xxx)),U_2(z\xxx))U_2'(z(\xxx))\big]\cdot\na z.
\end{equation*}
From (\ref{2_3}) we have by the maximum principle \cite{PW}, $0\leq z(\xxx)\leq 1$. On the other hand, for all $z\in [0,1]$ (\ref{1_6}) is true. Hence

\begin{equation*}
%\label{1_10}
\na\cdot\big[a_{11}(u_1(\xxx),u_2(\xxx))\na u_1+a_{12}(u_1(\xxx),u_2(\xxx))\na u_2\big]=0.
\end{equation*}
In the same way we prove (\ref{3_4}). Let now $(u_1(\xxx),u_2(\xxx))$ be a functional solution of (\ref{1_4})-(\ref{4_4}). Let $U_1(z)$, $U_2(z)$ and $z(\xxx)$ be the functions entering in the definition of functional solutions. We claim that there exist two constants $\g_1$ and $\g_2$ such that (\ref{1_6})-(\ref{3_6}) hold. (\ref{3_6}) follows immediately from the definitions involved. To prove (\ref{1_6}), let us define for $z\in[0,1]$ the function

\begin{equation}
\label{1_14}
\ttt_1(z)=\int_0^z[a_{11}(U_1(t),U_2(t))U'_1(t)+a_{12}(U_1(t),U_2(t))U'_2(t)]dt.
\end{equation}
We have, if $z$ is a function of $\xxx$,

\begin{equation*}
%\label{2_14}
\na\ttt_1(z(\xxx))=\big[a_{11}(U_1(z(\xxx)),U_2(z(\xxx))U'_1(z(\xxx))+a_{12}(U_1(z(\xxx)),U_2(z(\xxx)))U'_2(z(\xxx))\big]\na z.
\end{equation*}
Hence we obtain,  by (\ref{2_3}) and (\ref{1_4})

\begin{equation*}
%\label{3_14}
\D\ttt_1(z(\xxx))=0\quad\9.
\end{equation*}
If we define $\g_1=\ttt_1(1)$, $\ttt_1(z(\xxx))$ solves the problem

\begin{equation*}
%\label{1_15}
\D\ttt_1(z(\xxx))=0\quad\9
\end{equation*}

\begin{equation*}
%\label{2_15}
\ttt_1(z(\xxx))=0\quad\1,\quad\ttt_1(z(\xxx))=\g_1\quad\3
\end{equation*}

\begin{equation*}
%\label{3_15}
\frac{\pa\ttt_1(z(\xxx))}{\pa n}=\ttt'(z(\xxx))\frac{\pa z}{\pa n}=0\quad\2.
\end{equation*}
Thus we have

\begin{equation}
\label{4_15}
\ttt_1(z(\xxx))=\g_1 z(\xxx)\quad\hbox{in}\quad \bar\O.
\end{equation}
Let $\tilde z\in [0,1]$. A point $\tilde \xxx$ in $\O$ certainly exists such that $z(\tilde\xxx)=\tilde z$. Thus from (\ref{4_15}) we conclude that for all $z\in[0,1]$

\begin{equation*}
%\label{1_16}
\ttt_1(z)=\g_1 z.
\end{equation*}
From (\ref{1_14}) we have

\begin{equation*}
%\label{2_15}
\g_1 z=\int_0^z[a_{11}(U_1(t),U_2(t))U'_1(t)+a_{12}(U_1(t),U_2(t))U'_2(t)]dt.
\end{equation*}
Therefore

\begin{equation*}
%\label{2_15}
\g_1 =a_{11}(U_1(z),U_2(z))U'_1(z)+a_{12}(U_1(z),U_2(z))U'_2(z).
\end{equation*}
In the same way we find (\ref{2_6}).
\end{proof}

\section{``Small'' and ``large'' functional solutions }
In this Section we give a theorem of existence and uniqueness of small functional solutions of problem (\ref{1_4})-(\ref{4_4}) and a theorem of existence for solutions not necessarily small. We will find the small solutions with the corresponding small solutions of the two point problem 

\begin{equation}
\label{1_21}
a_{11}(U_1(z),U_2(z))\frac{dU_1}{dz}+a_{12}(U_1(z),U_2(z))\frac{dU_2}{dz}=\g_1
\end{equation}

\begin{equation}
\label{2_21}
a_{21}(U_1(z),U_2(z))\frac{dU_1}{dz}+a_{22}(U_1(z),U_2(z))\frac{dU_2}{dz}=\g_2
\end{equation} 

\begin{equation}
\label{3_21}
U_1(0)=0,\quad U_1(1)=u_1^*,\quad U_2(0)=0,\quad U_2(1)=u_2^*.
\end{equation}
We wish to prove that if $(u_1^*)^2+(u_2^*)^2$ is sufficiently small the problem (\ref{1_21})-(\ref{3_21}) has one and only one solution under the sole assumption

\begin{equation}
\label{6_21}
a_{11}(0,0)a_{22}(0,0)-a_{12}(0,0)a_{21}(0,0)\ne 0.
\end{equation}
To this end we use the inverse function theorem in Banach space which we quote below for the sake of completeness \cite{P}.

\begin{theorem}
Let $X$ and $Y$ be Banach spaces and $F(X)\in C^1(X,Y)$ a mapping from $X$ to $Y$. Assume $y^*=F(x^*),\ x^*\in X,\ y^*\in Y$ and assume the Frechet' s differential $F'(x^*)$ globally invertible as an application from $X$ to $Y$. Then there exist neighborhoods $U$ of $x^*$ and $V$ of $y^*$ such that $F:U\to V$ is an homeomorphism. Moreover $F^{-1}$ exists and $F^{-1}\in C^1(V,X)$ for all $y\in V$.
\end{theorem}
To use this Theorem we define the spaces

\begin{equation*}
%\label{1_25}
X=\bigl(C^1([0,1])\bigl)^2\times\RR,\quad Y=\bigl(C^0([0,1])\bigl)^2\times\RR\times\RR
\end{equation*}
and the mapping $F:X\to Y$

\begin{equation*}
%\label{2_25}
F((U_1,U_2),(\g_1,\g_2))=\bigl((a_{11}(U_1,U_2)U_1'+a_{12}(U_1,U_2)U_2'-\g_1,
\end{equation*}
\begin{equation*}
%\label{2_25}
a_{21}(U_1,U_2)U_1'+a_{22}(U_1,U_2)U_2'-\g_2), (U_1(0),U_2(0)),(U_1(1),U_2(1))\bigl),
\end{equation*}
where $(U_1(z),U_2(z))\in\bigl(C^1([0,1])\bigl)^2$ and $(\g_1,\g_2)\in\RR$.
For $x^*=((0,0),(0,0))\in X$ and $y^*=((0,0),(0,0),(0,0))\in Y$, we have $F(x^*)=y^*$. The differential of $F$ computed in $x^*$ and evaluated in $((\uu_1,\uu_2),(\G_1,\G_2))\in X$ is

\begin{equation*}
%\label{1_26}
F'(x^*)[(\uu_1,\uu_2),(\G_1,\G_2)]=
\end{equation*}
\begin{equation*}
%\label{2_25}
\bigl((a_{11}(0,0)\uu_1'+a_{12}(0,0)\uu_2'-\G_1,\ a_{21}(0,0)\uu_1'+a_{22}(0,0)\uu_2'-\G_2)),\ (\uu_1(0),\uu_2(0)),(\uu_1(1),\uu_2(1))\bigl).
\end{equation*}
\vskip .4 cm

We claim that for every $f=((G_1(z),G_2(z)),\ (\eta_1,\eta_2),\ (\xi_1,\xi_2))\in Y$ with $G_1(z),\ G_2(z)\in C^0([0,1])$, $(\eta_1,\eta_2)\in \RR$ and $(\xi_1,\xi_2)\in \RR$ the equation

\begin{equation}
\label{2_26}
F'(x^*)[(\uu_1,\uu_2),(\G_1,\G_2)]=f
\end{equation}
has one and only one solution. In components (\ref{2_26}) can be rewritten

\begin{equation}
\label{3_26}
a_{11}(0,0)\uu_1'+a_{12}(0,0)\uu_2'-\G_1=G_1(z)
\end{equation}

\begin{equation}
\label{4_26}
a_{21}(0,0)\uu_1'+a_{22}(0,0)\uu_2'-\G_2=G_2(z)
\end{equation}

\begin{equation}
\label{5_26}
\uu_1(0)=\eta_1,\ \uu_2(0)=\eta_2
\end{equation}

\begin{equation}
\label{6_26}
 \uu_1(1)=\xi_1,\ \uu_2(1)=\xi_2.
\end{equation}

By (\ref{6_21}) we can solve (\ref{3_26}), (\ref{4_26}) with respect to $\uu_1'(z)$, $\uu_2'(z)$. We obtain

\begin{equation*}
%\label{2_27}
\uu_1'(z)=\frac{1}{D}\big[\G_1a_{22}(0,0)-\G_2a_{12}(0,0)+a_{22}G_1(z)-a_{12}G_2(z)\big]
\end{equation*}

\begin{equation*}
%\label{3_27}
\uu_2'(z)=\frac{1}{D}\big[\G_2a_{11}(0,0)-\G_1a_{21}(0,0)+a_{11}G_2(z)-a_{21}G_1(z)\big],
\end{equation*}
where $D=a_{11}(0,0)a_{22}(0,0)-a_{12}(0,0)a_{21}(0,0)$.
In view of (\ref{5_26}) we have

\begin{equation*}
%\label{4_27}
\uu_1(z)=\frac{1}{D}\left[(\G_1a_{22}(0,0)-\G_2a_{12}(0,0))z+a_{22}\int_0^z G_1(t)dt-a_{12}\int_0^z G_2(t)dt\right]+\eta_1
\end{equation*}

\begin{equation*}
%\label{1_28}
\uu_2(z)=\frac{1}{D}\left[(\G_2a_{11}(0,0)-\G_1a_{21}(0,0))z+a_{11}\int_0^z G_2(t)dt-a_{21}\int_0^z G_1(t)dt\right]+\eta_2.
\end{equation*}
The conditions (\ref{6_26}) become

\begin{equation}
\label{2_28}
\xi_1=\frac{1}{D}\left[\G_1a_{22}(0,0)-\G_2a_{12}(0,0)+a_{22}\int_0^1 G_1(t)dt-a_{12}\int_0^1 G_2(t)dt\right]+\eta_1
\end{equation}

\begin{equation}
\label{3_28}
\xi_2=\frac{1}{D}\left[\G_2a_{11}(0,0)-\G_1a_{21}(0,0)+a_{11}\int_0^1 G_2(t)dt-a_{21}\int_0^1 G_1(t)dt\right]+\eta_2.
\end{equation}
By (\ref{6_21}) the linear system (\ref{2_28}), (\ref{3_28}) can be solved with respect to $\G_1$, $\G_2$. Hence the inverse function theorem is applicable. Moreover, not only the two point problem (\ref{1_6})-(\ref{3_6}) has one and only one solution if $(u_1^*)^2+(u_2^*)^2$ is sufficiently small, but the same is true for the boundary value problem (\ref{1_4})-(\ref{4_4}) in the class of functional solutions.
\bigskip

A general question implied in the point of view of functional solution concerns when every classical solution is also a functional solution, i.e. if $\cc_F=\cc$. The answer is  certainly positive for the small solutions. For, we can apply the inverse function Theorem 3.1 directly to the problem (\ref{1_4})-(\ref{4_4}) and prove that there is a unique small solution. Another example in which $\cc_F=\cc$ is given in Section 4.
\bigskip

We prove now a theorem of existence for the problem (\ref{1_6})-(\ref{3_6}) which, in turn, will imply the existence of functional solutions for (\ref{1_4})-(\ref{4_4}).

\begin{theorem}
Let us assume  $a_{ij}(U_1,U_2)=a_{ji}(U_1,U_2)$ and, for all $\xi=(\xi_1,\xi_2)\in\RR$,

\begin{equation}
\label{1_35}
 M|\xi|^2\geq\sum_{ij=1}^2 a_{ij}(U_1,U_2)\xi_i\xi_j\geq m|\xi|^2
\end{equation}
where $M>0$ and $m>0$ are two constants not depending on $(U_1,U_2)$. Then the problem

\begin{equation}
\label{2_35}
 A(\UUU)\UUU'=\g, \quad A=(a_{ij}(\UUU)),\ \UUU=(U_1,U_2),\ \UUU'=(U_1',U_2'),\ \g=(\g_1,\g_2)
\end{equation}

\begin{equation}
\label{3_35}
 \UUU(0)=(0,0),\quad \UUU(1)=\uuu^*,\quad \uuu^*=(u_1^*,u_2^*)
\end{equation}
 has at least one solution.
\end{theorem}  

\begin{proof}
By (\ref{1_35}) the matrix $A$ is invertible. Thus we can rewrite the system (\ref{2_35}), (\ref{3_35}) as follows

\begin{equation}
\label{2_36}
 \UUU'(z)=(A^{-1}\UUU(z))\g,\quad  \UUU(0)=(0,0),\quad \UUU(1)=\uuu^*.
\end{equation}
Let $(\UUU(z),\g)$ be any solution of (\ref{2_36}). Integrating (\ref{2_36}) from $0$ to $1$ we have

\begin{equation}
\label{3_36}
\uuu^* =\left(\int_0^1 A^{-1}\UUU(t) dt\right)\g.
\end{equation}
We wish to solve (\ref{3_36}) with respect to $\g$. To this end we study the invertibility of the matrix

\begin{equation*}
%\label{1_37}
\int_0^1 A^{-1}\UUU(t) dt.
\end{equation*}
Let $\l_m(\UUU)$ and $\l_M(\UUU)$ be the eigenvalues of $A(\UUU)$. By (\ref{1_35}) we have

\begin{equation*}
%\label{2_37}
0<m\leq\l_m(\UUU(z))\leq\l_M(\UUU(z))\leq M.
\end{equation*}
Denote by $e_{ij}(\UUU)$ the elements of $A^{-1}(\UUU)$. By (\ref{1_35}) we have, for all $\xi\in\RR$,

\begin{equation}
\label{5_37}
\sum_{ij=1}^2 e_{ij}(\UUU)\xi_i\xi_j\geq\frac{|\xi|^2}{\l_M(\UUU)}\geq\frac{|\xi|^2}{M}.
\end{equation}
Therefore

\begin{equation*}
%\label{1_38}
\sum_{ij=1}^2\Bigl(\int_0^1 e_{ij}(\UUU(t))dt\Bigl)\xi_i\xi_j\geq\frac{|\xi|^2}{M}.
\end{equation*}
Thus $\int_0^1 A^{-1}(\UUU(t))dt$ is invertible and we obtain $\g$ as a vector functionally dependent on $\UUU(z)$ i.e

\begin{equation*}
%\label{2_38}
\g[\UUU]=\Bigl(\int_0^1 A^{-1}(\UUU(t))dt\Bigl)^{-1}\uuu^*.
\end{equation*}
Hence the problem (\ref{2_35}), (\ref{3_35}) can be rewritten as

\begin{equation}
\label{2_39}
\UUU(z)=\Bigl(\int_0^z A^{-1}(\UUU(t))dt\Bigl)\Bigl(\int_0^1 A^{-1}(\UUU(t))dt\Bigl)^{-1}\uuu^*.
\end{equation}
To prove that (\ref{2_39}) has at least one solution we use the Schauder's fixed point theorem. Since $m$ and $M$ do not depend on $\UUU$, $\g[\UUU]$ is also bounded by a constant $C$. On the other hand, also the norm of the matrix

\begin{equation*}
%\label{2_38}
\int_0^z A^{-1}(\UUU(t))dt,\quad z\in[0,1]
\end{equation*} 
can be majorized by a constant $C_1$ depending only on the data. Hence we have the a priori estimate

\begin{equation*}
%\label{2_40}
\pp \UUU(z)\pp_{C^0([0,1])}\leq C_2.
\end{equation*} 
By (\ref{2_36}) also $\UUU'(z)$ is a priori bounded in the $C^0([0,1])$-norm. Let

\begin{equation*}
%\label{3_40}
\bb=\{\UUU(z);\UUU(z)\in C^0([0,1]),\UUU(0)=0, \UUU(z)=\uuu^*,\pp\UUU\pp_{C^0([0,1])}\leq C_2\}.
\end{equation*} 
$\bb$ is closed, bounded and convex subset of $C^0([0,1])$. We define in $\bb$ the operator $T[\UUU]$

\begin{equation*}
%\label{1_41}
T[\UUU](z)=\Bigl(\int_0^z A^{-1}(\UUU(t))dt\Bigl)\Bigl(\int_0^1 A^{-1}(\UUU(t))dt\Bigl)^{-1}\uuu^*.
\end{equation*} 
We have $T(\bb)\subseteq\bb$. Moreover, if $\UUU\in\bb$, we have $\pp T[\UUU]'(z)\pp_{C^0([0,1])}\leq C_3$. Thus by Arzela's theorem $T$ is a compact operator. We conclude that $T$ has at least a fixed point which gives a solution of the problem (\ref{1_6})-(\ref{3_6}).
\end{proof}

\section{Functional solutions for a fluid motion obeying the Darcy's law}
Let $\O$, $\G_1$, $\G_2$ and $\G_3$ be as in Section 2 and 3. We treat here the case of a macroscopic and molecular motion of an incompressible fluid occurring in the porous medium $\O$. The motion obeys the Darcy's law \cite{JB}

\begin{equation}
\label{1_42}
\vvv=-K\na p,
\end{equation} 
where $p$ is the pressure and $\vvv$ the velocity. For the heat and mass flow densities $\qqq_h$ and $\qqq_m$ we have, taking into account the Soret and Dofour effects \cite{NB},

\begin{equation*}
%\label{1_43}
\qqq_h=a_{11}\na u_1+a_{12}\na u_2+u_1\vvv,\quad  \qqq_m=a_{21}\na u_1+a_{22}\na u_2+u_2\vvv,
\end{equation*}
where $u_1(\xxx)$, $u_2(\xxx)$ are the temperature and concentration and the $a_{ij}$ are assumed to be given functions of $u_1$, $u_2$ and $p$. From the conditions

\begin{equation*}
%\label{3_43}
 \na\cdot\qqq_h=0,\quad \na\cdot\qqq_m=0,\quad\na\cdot\vvv=0,
\end{equation*}
using (\ref{1_42}), we arrive at the problem

\begin{equation}
\label{4_43}
\na\cdot\big[a_{11}(u_1,u_2,p)\na u_1+a_{12}(u_1,u_2,p)\na u_2+u_1K(u_1,u_2,p)\na p\big]=0\quad\9
\end{equation}

\begin{equation}
\label{7_43}
u_1=0\ \1,\ \frac{\pa u_1}{\pa n}=0\ \2,\ u_1=u_1^*\ \3
\end{equation}

\begin{equation}
\label{5_43}
\na\cdot\big[a_{21}(u_1,u_2,p)\na u_1+a_{22}(u_1,u_2,p)\na u_2+u_2K(u_1,u_2,p)\na p\big]=0\quad\9
\end{equation}

\begin{equation}
\label{8_43}
u_2=0\ \1,\ \frac{\pa u_2}{\pa n}=0\ \2,\ u_2=u_2^*\ \3
\end{equation}

\begin{equation}
\label{6_43}
\na\cdot(K(u_1,u_2,p)\na p)=0
\end{equation}

\begin{equation}
\label{9_43}
p=0\ \1,\ \frac{\pa p}{\pa n}=0\ \2,\ p=p^*\ \3.
\end{equation}
The problem (\ref{4_43})-(\ref{9_43}) is a special case of the following more general problem \cite{GC}, \cite{GC1}

\begin{equation}
\label{1_44}
\na\cdot\Bigl[\sum_{j=1}^n a_{ij}(u_1,...u_n,p)\na u_j+b_i(u_1,....u_n,p)\na p\Bigl]=0\quad\9,\quad i=1....n
\end{equation}

\begin{equation}
\label{2_44}
\na\cdot\Bigl[b_{n+1}(u_1,....u_n,p)\na p\Bigl]=0\quad\9
\end{equation}

\begin{equation}
\label{3_44}
u_i=0\ \1,\ \frac{\pa u_i}{\pa n}=0\ \2, u_i=u_i^*\ \3, \ i=1....n
\end{equation}

\begin{equation}
\label{4_44}
 p=0\ \1,\  \frac{\pa p}{\pa n}=0\2,\ p=p^*\ \3,
\end{equation}
where now $\O$ is an open and bounded subset of $\RN$ and $\G_1$, $\G_2$ and $\G_3$ are defined as in the case of Section 2. Moreover, $u_i^*$, $p^*$ are given constants. To treat the problem (\ref{1_44})-(\ref{4_44}) from the point of view of functional solutions we assume, in this section, a perspective slightly different from that of Section 2. There the equation and the boundary conditions  (\ref{1_5}) determining the ``pivot'' $z(\xxx)$ were artificially added to the problem. Now we take as ``pivot'' the pressure $p$ i.e. one of the unknown of the problem(\ref{1_44})-(\ref{4_44}). Again with the goal to separate the geometric part of the problem from the nonlinear part. To this end we give the following new definition.

\begin{definition}
A solution $(u_1(\xxx),...,u_n(\xxx),p(\xxx))$ of problem (\ref{1_44})-(\ref{4_44}) is termed ``functional'' if $n$ functions $U_1(p),...,U_n(p)$, each of class $C^1([0,p^*])$, exist such that

\begin{equation*}
%\label{3_45}
u_i(\xxx)=U_i(p(\xxx)),\quad i=1...n.
\end{equation*}
\end{definition}
Clearly, the set $\cc_F$ of functional solutions is a subset of the set $\cc$ of all classical solutions. It is unclear to the writer if the vice-versa is also true. The two point problem associated with (\ref{1_44})-(\ref{4_44}) is

\begin{equation}
\label{1_47}
\sum_{j=1}^n a_{ij}(U_1,...,U_n,p)\frac{dU_j}{dp}+b_i(U_1,...,U_n,p)=\g_i b_{n+1}(U_1,...,U_n,p),\ i=1..n
\end{equation}

\begin{equation}
\label{2_47}
U_i(0)=0,\quad U_i(p^*)=u_i^*.
\end{equation}
The following theorem relates the solutions of the boundary problem (\ref{1_44})-(\ref{4_44}) to the solutions of the two-point problem (\ref{1_47}), (\ref{2_47}). For the proof we refer to \cite{GC}.

\begin{theorem}
Let $a(U_1,...,U_n,p)$, $b_i(U_1,...,U_n,p)$ and $b_{n+1}(U_1,...,U_n,p)$ belong to $C^0([\RN\times\R)$ and
\begin{equation}
\label{3_47}
b_{n+1}(U_1,...U_n,p)\geq b_0\geq 0\quad \hbox{for all}\ (U_1,...,U_n)\in\RN,\ \hbox{and for all}\ p\in\R.
\end{equation}
Then to every solution $(U_1,(p)...,U_n(p))$, $(\g_1,...,\g_n)$ of the problem (\ref{1_47}), (\ref{2_47}) there corresponds a solution of the problem (\ref{1_44})-(\ref{4_44}) belonging to $\cc_F$. Vice-versa for every functional solutions of (\ref{1_44})-(\ref{4_44}) we have a solution of the two point problem (\ref{1_47}), (\ref{2_47}).
\end{theorem}
\noindent We have also

\begin{theorem}
Let the same assumptions of Theorem 4.2 hold for on  $a(U_1,...,U_n,p)$, $b_i(U_1,...,U_n,p)$ and $b_{n+1}(U_1,...,U_n,p)$. then every functional solution of problem (\ref{1_44})-(\ref{4_44}) can be expressed in terms of the solution $U_i(p)$ of the corresponding two point problem (\ref{1_47}), (\ref{2_47}) and of the solution $z(\xxx)$ of the mixed problem for the laplacian

\begin{equation}
\label{1_48}
\D z=0\ \9,\quad z=0\ \1,\quad\frac{\pa z}{\pa n}=0\ \2,\quad z=1\ \3.
\end{equation}
\end{theorem}

\begin{proof}
Let $(u_1(\xxx),...,u_n(\xxx),p(\xxx))=(U_1(p(\xxx)),...,U_n(p(\xxx)),p(\xxx))$ be a functional solution of (\ref{1_44})-(\ref{4_44}). We have by (\ref{2_44})

\begin{equation}
\label{1_49}
\na\cdot\Bigl[b_{n+1}(U_1(p(\xxx)),...,U_n(p(\xxx)),p(\xxx))\na p\Bigl]=0\quad\9
\end{equation}

\begin{equation}
\label{2_49}
 p=0\ \1,\  \frac{\pa p}{\pa n}=0\2,\ p=p^*\ \3, \quad i=1....n.
\end{equation}
The problem (\ref{1_49}), (\ref{2_49}) can be solved with the aid  of the Kirchhoff's transformation. More precisely, let us define the function

\begin{equation*}
%\label{3_49}
\eta=\Theta(p)\quad \hbox{where}\quad \Theta(p)=\int_0^p b_{n+1}(U_1(t),...,U_n(t),t)dt.
\end{equation*}
Since $\na\eta=b_{n+1}(U_1(p(\xxx)),...,U_n(p(\xxx)),p(\xxx))\na p$, $\eta(\xxx)$ is a solution of the linear problem

\begin{equation*}
%\label{1_50}
\D \eta=0\ \9,\ \eta=0\ \1,\ \frac{\pa \eta}{\pa n}=0\ \2,\ \eta=\eta^*\ \3,\ \hbox{where}\ \eta^*=\Theta(p^*).
\end{equation*}
$\eta(\xxx)$ is immediately obtained in term of $z(\xxx)$. For, in view of the uniqueness which holds for the problem (\ref{1_48}), we have

\begin{equation}
\label{4_50}
\eta(\xxx)=\eta^* z(\xxx).
\end{equation}
By (\ref{3_47}) $\Theta(p)$ maps one-to-one $[0,p^*]$ onto $[0,\eta^*]$. Therefore, by (\ref{4_50}) we have

\begin{equation}
\label{5_50}
p(\xxx)=\Theta(\eta(\xxx))=\Theta(\eta^*z(\xxx)).
\end{equation}
Thus the functional solution $(u_1(\xxx),...,u_n(\xxx),p(\xxx))=(U_1(p(\xxx)),...,U_n(p(\xxx)),p(\xxx))$ is fully determined.
\end{proof}
The question if $\cc_F=\cc$ can be settled in the following special case of problem (\ref{1_44})-(\ref{4_44}):

\begin{equation}
\label{1_51}
\na\cdot(a(u,p)\na u)=0\quad\9
\end{equation}

\begin{equation}
\label{2_51}
\na\cdot(a(u,p)\na p)=0\quad\9
\end{equation}

\begin{equation}
\label{3_51}
u=0\ \1,\quad\frac{\pa u}{\pa n}=0\ \2,\quad u=u^*\ \3
\end{equation}

\begin{equation}
\label{4_51}
p=0\ \1,\quad\frac{\pa p}{\pa n}=0\ \2,\quad p=p^*\ \3.
\end{equation}

\begin{theorem}
Assume $a(u,p)\in C^0(\RR)$ and

\begin{equation}
\label{0_52}
a(u,p)\geq\a>0.
\end{equation}
The problem (\ref{1_51})-(\ref{4_51}) has one and only one classical solution which is also a functional solution. Therefore, in this case $\cc_F=\cc$.
\end{theorem}

\begin{proof}
Let $(u(\xxx),p(\xxx))$ be a classical solution of (\ref{1_51})-(\ref{4_51}). Define

\begin{equation*}
%\label{1_52}
\z(\xxx)=u(\xxx)-\frac{u^*}{p^*}p(\xxx).
\end{equation*}
We claim that $\z(\xxx)=0\ \9$. We have by (\ref{1_51}) and (\ref{2_51})

\begin{equation}
\label{1_53}
\na\cdot(a(u,p)\na \z)=\na\cdot(a(u,p)\na u)-\frac{u^*}{p^*}\na\cdot(a(u,p)\na p)=0
\end{equation}
and

\begin{equation}
\label{3_52}
\z=0\ \1,\quad\frac{\pa \z}{\pa n}=0\ \2,\quad \z=0\ \3.
\end{equation}
Multiplying (\ref{1_53}) by $\z(\xxx)$, integrating by parts over $\O$ and taking into account the boundary condition (\ref{3_52}) we obtain

\begin{equation*}
%\label{4_53}
\int_\O a(u,p)|\na\z|^2dx=0.
\end{equation*}
Hence, by (\ref{0_52}) we have $\z(\xxx)=0\ \9$ and $u(\xxx)=\frac{u^*}{p^*}p(\xxx)$. Thus $(u(\xxx),p(\xxx))$ is the only solution of (\ref{2_51})-(\ref{4_51}) and is a functional solution with

\begin{equation}
\label{4_53}
U(p)=\frac{u^*}{p^*}p.
\end{equation}  
\end{proof}

We note that we may compute the solution of problem (\ref{1_51})-(\ref{4_51}) using Theorem 4.2. In this simple situation the two point problem (\ref{1_47}), (\ref{2_47}) reduces to

\begin{equation*}
\frac{dU}{dp}=\g,\quad U(0)=0,\quad U(p^*)=u^*
\end{equation*}  
whose solution is precisely (\ref{4_53}).

The problem (\ref{1_44})-(\ref{4_44}) can have more than one  functional solution
 as in the following example:

\begin{equation}
\label{1_57}
\na\cdot(\na u_1-u_2\na p)=0\ \quad \9
\end{equation}  

\begin{equation}
\label{2_57}
\na\cdot(\na u_2-u_1\na p)=0\ \quad \9
\end{equation} 

\begin{equation}
\label{3_57}
\na\cdot(\na p)=0\ \quad \9
\end{equation}  

\begin{equation}
\label{4_57}
u_1=0,\quad u_2=0,\quad p=0\quad \1
\end{equation} 

\begin{equation}
\label{5_57}
\frac{\pa u_1}{\pa n}=0,\quad \frac{\pa u_2}{\pa n}=0,\quad \frac{\pa p}{\pa n}=0\quad \2
\end{equation}     

\begin{equation}
\label{6_57}
u_1=0,\quad u_2=0,\quad p=0\quad \3.
\end{equation}   
According to Theorem 4.2 the corresponding two point problem is 

\begin{equation}
\label{1_58}
\frac{dU_1}{dp}=\g_1+U_2,\quad U_1(0)=0,\quad U_1(p^*)=0
\end{equation}   

\begin{equation}
\label{2_58}
\frac{dU_2}{dp}=\g_2-U_2,\quad U_2(0)=0,\quad U_2(p^*)=0
\end{equation}   
which can easily be solved. We find

\begin{equation*}
%\label{3_58}
U_1(p)=\g_1\sin p+\g_2(1-\cos p),\quad U_2(p)=\g_1(\cos p-1)+\g_2\sin p,
\end{equation*}   
\vskip .3 cm where $(\g_1,\g_2)$ is the solution of the linear system

$$
\biggl\{
\begin{array}{ccc}
\g_1\sin p^*+\g_2(1-\cos p^*)=0\\
\g_1(\cos p^*-1)+\g_2\sin p^*=0.\\
\end{array}
$$
\vskip .3 cm
If $1-\cos p^*\ne 0$ the problem (\ref{1_58}), (\ref{2_58}) has only the trivial solution $U_1(p)=0,\quad U_2(p)=0,\quad \g_1=0,\quad \g_2=0 $ and, correspondingly, for the problem (\ref{1_57})-(\ref{6_57}) we have $u_1(\xxx)=0,\quad u_2(\xxx)=0,\quad p(\xxx)=0$. On the other hand, if $1-\cos p^*=0$ we have for (\ref{1_57})-(\ref{6_57}) and (\ref{1_58}), (\ref{2_58}) an infinite number of solutions which are easily computed.

 To show a case of problem (\ref{1_44})-(\ref{4_44}) for which we have uniqueness of functional solutions we use the following

\begin{theorem}
Let $F(U,p)$ be measurable with respect to $p$ and continuous with respect to $U$ in the rectangle $R=\{(p,U);\ \leq p\leq A,\ 0\leq U\leq B\}$. Assume that there exist two functions $q(p)$ and $r(p)$ such that

\begin{equation*}
%\label{1_62}
r(p)\leq F(U,p)\leq q(p)
\end{equation*}

\begin{equation*}
%\label{2_62}
r(p)\geq 0,\quad \int_0^A r(t)dt>0.
\end{equation*}
Then the problem

\begin{equation}
\label{3_62}
\frac{dU}{dp}(p)=\g F(U(p),p),\quad U(0)=0, \quad U(A)=0
\end{equation}
has at least one solution absolutely continuous in $[0,A]$. Moreover, if $F(U,p)$ satisfies a Lipschitz condition in $R$
 with respect to $U$ the solution of (\ref{3_62}) is unique.
\end{theorem}
For the proof of Theorem 4.5 and for the related topics which are closely connected to the present work we refer to \cite{ZW}, \cite{ZW1}, \cite{C}, \cite{HN}, \cite{Z} and to the book \cite{S1} page 105. We apply the above Theorem to prove

\begin{theorem}
If $a(u,p)\in C^0(R)$, $b(u,p)\in C^0(R)$, $a(u,p)>0$ in $R$, $b(u,p)>0$ in $R$ and $\frac{b(u,p)}{a(u,p)}$ is uniformly Lipschitz with respect to $u$ in $R$ then the problem

\begin{equation}
\label{1_60}
\na\cdot(a(u,p)\na u)=0\quad\9
\end{equation}

\begin{equation}
\label{2_60}
\na\cdot(b(u,p)\na p)=0\quad\9
\end{equation}

\begin{equation}
\label{3_60}
u=0\ \1,\quad\frac{\pa u}{\pa n}=0\ \2,\quad u=u^*\ \3
\end{equation}

\begin{equation}
\label{4_60}
p=0\ \1,\quad\frac{\pa p}{\pa n}=0\ \2,\quad p=p^*\ \3
\end{equation}
has one and only one functional solution.
\end{theorem}

\begin{proof}
The two point problem associated with (\ref{1_60})-(\ref{4_60}) is

\begin{equation}
\label{2_63}
\frac{dP}{du}=\g\frac{b(u,P)}{a(u,P}
\end{equation}

\begin{equation}
\label{3_63}
U(0)=0,\quad U(p^*)=u^*.
\end{equation}

The Theorem 4.5 is applicable to the two-point problem (\ref{2_63}), (\ref{3_63}) if we define $F(u,p)=\frac{b(u,p}{a(u,p)}$. We conclude that problem (\ref{2_63})-(\ref{3_63}) has one and only one solution and by Theorem 4.2 the same is true for the problem (\ref{1_60})-(\ref{4_60})  in the class of functional solutions $\cc_F$.
\end{proof}

 Note that in (\ref{1_60})-(\ref{4_60}) we can equivalently take as pivot $p$ or $u$. In this second case the relevant two point problem becomes

\begin{equation*}
%\label{2_63}
\frac{dP}{du}=\g\frac{a(u,P)}{b(u,P}
\end{equation*}

\begin{equation*}
%\label{3_63}
P(0)=0,\quad P(u^*)=p^*.
\end{equation*}

%\noindent{\bf Compliance with ethical standard}
%\vskip .3cm
%\noindent{\bf Conflict of interest.} The author declares that he has no conflicts of interest.  

\bibliographystyle{amsplain}

\end{document}